\documentclass[a4paper, twoside,12pt]{article}
\usepackage{fancyhdr}
\usepackage{fnpos}
\usepackage[english]{babel}        
\usepackage[T1]{fontenc}          
\usepackage{graphicx}             
\usepackage{makeidx}
\usepackage{fancybox}
\usepackage{framed}
\usepackage{fancyhdr}
\usepackage{pstricks,pst-plot,pstricks-add}
\usepackage[margin=1in]{geometry}
\usepackage{graphicx}
\usepackage{titlesec}
\usepackage{amsmath}
\usepackage{amsfonts}   
\usepackage{amssymb}    
\usepackage{amsthm}
\usepackage{dsfont}
\usepackage{mathtools}
\usepackage{verbatim}   
\usepackage{color}
\usepackage{listings}
\usepackage{graphicx}
\usepackage{amsmath}
\usepackage{amsmath,amssymb,color,inputenc,euscript,graphicx,psfrag}

\newtheorem{theorem}{Theorem}[section]
\newtheorem{corollary}[theorem]{Corollary}
\newtheorem{lemma}[theorem]{Lemma}
\newtheorem{proposition}[theorem]{Proposition}
\newtheorem{definition}[theorem]{Definition}
\newtheorem{properties}[theorem]{Properties}
\newtheorem{remark}[theorem]{Remark}

\numberwithin{equation}{section}

\DeclareMathOperator*{\supess}{sup\,ess}

\usepackage{amsmath, amsthm, amscd, amsfonts, amssymb, graphicx, color,   mathrsfs}
\usepackage[english]{babel}
\usepackage[bookmarksnumbered, colorlinks, plainpages]{hyperref}
\hypersetup{colorlinks=true,linkcolor=blue, anchorcolor=blue, citecolor=blue, urlcolor=red, filecolor=magenta, pdftoolbar=true}

\everymath{\displaystyle}
\usepackage{orcidlink} 

\usepackage{bm}

 {\markboth{\sl A. Saoudi }{\sl Quadratic-Phase Fourier--Bessel Transform}
	
	\author {Ahmed Saoudi $^{ \orcidlink{0000-0003-4422-2054}}$}}
	
	\title{Quadratic-Phase Fourier--Bessel Transform: definitions, properties and uncertainty principles}
	
	\date{}

\begin{document}
		\maketitle
		\begin{center}
	Department of Mathematics, College of Science, Northern Border University, Arar, Saudi Arabia\\
			Email: ahmed.saoudi@ipeim.rnu.tn\\
			
		\end{center}
			\begin{abstract}
					In this manuscript, we introduce the quadratic-phase Fourier–Bessel transform and develop its foundational properties, including continuity, the Riemann–Lebesgue lemma, reversibility, and Parseval's identity. We define the associated translation operator and convolution product, establishing their main properties within this framework. As an application, we prove a Donoho–Stark-type uncertainty principle for the quadratic-phase Fourier–Bessel transform, extending classical uncertainty results to this generalized setting. 
				\end{abstract}
		
		\section{Introduction}
Over the past few decades, the classical Fourier transform has been extended and generalized in several directions, leading to important operators such as the fractional Fourier transform \cite{bahba2024fractional, condon1937immersion,  haouala2022fractional, mcbride1987namias, wiener1929hermitian},the linear canonical transform \cite{collins1970lens, haouala2025linear, saoudi2024hardy, wiener1929hermitian, wolf1974canonical} and many others. More recently, Castro et al, introduced the quadratic phase Fourier transform as a further generalization. This transform is defined through an exponential kernel with a quadratic phase, and it unifies and extends several well-known integral transforms

The quadratic phase Fourier transform is a powerful generalization of the classical Fourier transform, designed to address the limitations of traditional Fourier analysis in handling non-stationary signals. Unlike the Fourier transform, which assumes that signals are stationary (i.e., their frequency content does not change over time), the     quadratic phase Fourier transform incorporates quadratic-phase terms, making it particularly effective for analyzing signals with time-varying frequency content, such as chirp signals, radar pulses, and optical wavefronts.

The     quadratic phase Fourier transform of a function \( f(t) \) is defined as:

\[
\mathcal{Q}^{a,b,c,}_{d,e}[f](u) = \int_{-\infty}^\infty K^{a,b,c,}_{d,e}(t, u) f(t) \, dt,
\]
where the kernel \( K^{a,b,c,}_{d,e}(t, u) \) is given by:

\[
K^{a,b,c,}_{d,e}(t, u) = \frac{1}{\sqrt{2i\pi b}} \, e^{i\left(at^2 + bt u + cu^2 + dt + eu\right)},
\]
with, \( a, b, c, d, e \) are real parameters that control the quadratic and linear phase terms. The parameter \( b \neq 0 \) ensures the transform is well-defined.

The quadratic phase Fourier transform was developed and studied in various setting, including wavelet transform \cite{prasad2020quadratic, shah2022quadratic, sharma2023quadratic}, uncertainty principles \cite{castro2023uncertainty, castro2025multidimensional, shah2021uncertainty},  quaternion domains 
\cite{ahmad2025short, bhat2024novel, gargouri2024novel, gupta2024short,  zayed2025discrete,  dar2207towards}, and octonion domains \cite{kumar2024octonion, lone2025analysis}.

The  quadratic phase Fourier transform was introduced as a unifying framework that generalizes several well-known transforms, including the Fourier transform, the fractional Fourier transform, and the linear canonical transform. By incorporating additional parameters that control quadratic and linear phase terms, the     quadratic phase Fourier transform provides greater flexibility in signal analysis and has found applications in diverse fields such as signal processing, optics, quantum mechanics, and image processing.

In the present article, we introduce and study the  quadratic-phase Fourier--Bessel  transform that depend on five real parameters \( a, b, c, d, e \) and the indice $ \gamma > -1/2$, where $b\neq 0$. We denote this new transformation by $\mathbb{B}^{a,b,c}_{d,e,\gamma}$ which extends a broad range of integral transform families by selecting appropriate parameters.

We define the  quadratic-phase Fourier--Bessel  transform of an integrable function $h$ on the half line as 

\begin{equation}\label{QF.FBT}  
	\mathbb{B}^{a,b,c}_{d,e,\gamma}[h](t)= \frac{c_\gamma}{( i b)^{\gamma+1}} \int_{0}^\infty \mathbb{J}^{a,b,c}_{d,e,\gamma}(t,s) h(s)s^{2 \gamma+1} d s , \quad b \neq 0,
\end{equation}
where $\mathbb{J}^{a,b,c}_{d,e,\gamma}(t,s)$ denotes the quadratic-phase Fourier--Bessel kernel,  given by

\begin{equation}\label{QF.FBTK}
	\mathbb{J}^{a,b,c}_{d,e,\gamma}(t,s)= e^{-i(as^2+ct^2+ds+et)} j_{\gamma}\left(st/b\right).
\end{equation}
where $j_{\gamma}$ represent the normalized  Bessel function given by
\begin{equation}\label{seriesb}
	j_{\gamma}(s):=2^{\gamma} \Gamma(\gamma+1) \frac{J_{\gamma}(s)}{s^{\gamma}}=\Gamma(\gamma+1) \sum_{n=0}^{+\infty} \frac{(-1)^{n}(s / 2)^{2 n}}{n ! \Gamma(n+\gamma+1)},
\end{equation}
and the constant $	c_\gamma$ is given by
\begin{equation}\label{constc_gamma}
	c_\gamma= \frac{1}{2^{\gamma}\Gamma(\gamma+1)}.
\end{equation}
Here $J_{\gamma}$ denotes the Bessel function (see \cite{Watson1944}).

With suitable choices of parameters  \( a, b, c, d, e \) and multiplicity function $\gamma$, the  quadratic-phase Fourier--Bessel  transform (\ref{QF.FBT}) reduces to several well-known integral transforms, as shown below:
\begin{itemize}

 \item If $d=e=0$, $b\neq 0$,  and under  the transformations $a \mapsto - a/2b$ and $c \mapsto - c/2b$, then quadratic-phase Fourier--Bessel  transform  reduced to the  linear canonical Fourier--Bessel  transform   \cite{ dhaouadi2020harmonic, mohamed2024linear}, defined  by

\begin{equation}\label{LCFBT}
	\mathcal{L}^\mathcal{N}_{\gamma} [h](t)= \frac{c_\gamma}{( i b)^{\gamma+1}}\int_{0}^\infty \mathscr{K}_{\gamma}(t,s) h(s)s^{2 \gamma+1} d s,\quad   b \neq 0,
\end{equation}
where $\mathscr{K}_{\gamma}(t,s)$ represents the linear canonical Fourier--Bessel  kernel defined by 
\begin{equation}\label{LCFBT.K}
	 \mathscr{K}_{\gamma}(t,s)=e^{\frac{i}{2}\left(\frac{a}{b} s^{2}+\frac{d}{b} t^{2}\right)} j_{\gamma}\left(st/b\right).
\end{equation}

 \item If $d=e=0$,   and under  the transformations  $a \mapsto - \cot\theta/2$, $b=\sin\theta$ and $c \mapsto - \cot\theta/2$ with $\theta\in\mathbb{R}_+\textbackslash\pi\mathbb{Z}$, then quadratic-phase Fourier--Bessel  transform  reduced to the  fractional Fourier--Bessel  transform\cite{kerr1991fractional}, defined  by

\begin{equation}\label{frFBT}
	\mathcal{L}^{\theta}_{\gamma} [h](t)= \frac{c_\gamma c_\theta}{|\sin\theta|^{\gamma+1}}\int_{0}^\infty \mathscr{J}_{\gamma}(t,s) h(s)s^{2 \gamma+1} d s,
\end{equation}
where $\mathscr{J}_{\gamma}(t,s)$ represents the linear canonical Fourier--Bessel  kernel defined by 
\begin{equation}\label{frFBT.K}
	\mathscr{J}_{\gamma}(t,s)=e^{-\frac{i}{2} \cot (\theta)\left(s^2+t^2\right)}  j_{\gamma}\left(st/\sin (\theta)\right),
\end{equation}
and 
\begin{equation*}
    c_\theta =e^{i(\gamma+1)((\theta-2 n \pi)-\hat{\theta} \pi / 2)}, \quad \textrm{with}\quad \hat{\theta}=sgn(\sin\theta).
\end{equation*}

 \item If $a=c=d=e=0$,   and $b=1$, then the quadratic-phase Fourier--Bessel  transform  reduced to the   Fourier--Bessel  transform, except for a constant unimodular factor \( e^{-i(\mu+1) \pi/2} \), given by

\begin{equation}\label{FBT}
	\mathcal{B}_{\gamma} [h](t)= c_\gamma\int_{0}^\infty j_{\gamma}(st) h(s)s^{2 \gamma+1} d s,
\end{equation}
where the kernel $j_{\gamma}$ is given by (\ref{seriesb}).
\end{itemize}

In this work, we extend the framework of quadratic-phase Fourier analysis by introducing the quadratic-phase Fourier–Bessel transform, a new integral transform that combines quadratic phase modulation with the Bessel structure. We develop its foundational properties, including continuity, the Riemann–Lebesgue lemma, reversibility, and Parseval’s identity, which ensure that the transform is well defined and suitable for harmonic analysis. To build an operational calculus, we introduce the associated translation operator and convolution product, and we establish their main properties in this generalized setting. These results provide the analytic tools necessary for studying uncertainty principles in the new domain. As an application, we prove a Donoho–Stark-type uncertainty principle for the quadratic-phase Fourier–Bessel transform, thereby extending classical uncertainty results to a broader and more flexible transform framework.

The paper is organized as follows. In Section \ref{sec2}, we introduce the quadratic-phase Fourier–Bessel transform and establish its fundamental properties. Section \ref{sec3} is devoted to the definition of generalized translation operators associated with this transform and to the study of their main properties. In Section \ref{sec4}, we define the generalized convolution product in the quadratic-phase Fourier–Bessel setting and investigate its properties. Finally, in Section \ref{sec5}, we apply the developed framework to prove a Donoho–Stark-type uncertainty principle for the quadratic-phase Fourier–Bessel transform.

\section{Fundamental properties of  quadratic-phase Fourier--Bessel  transform}\label{sec2}

Along this article, We consider the following notations and functional spaces.
\begin{itemize}
	\item   \( \mathcal{C_{*}}(\mathbb{R}) \) denotes the space consisting of all continuous functions on \( \mathbb{R} \), and endowed with the topology of uniform convergence on compact subsets.
	
	\item  \( \mathcal{C}_{*, 0}(\mathbb{R}) \) is the set of all even continuous functions on \( \mathbb{R} \) that approach zero at infinity,  equipped with the topology of uniform convergence.
	
	\item  $L^{p}_{\gamma}(\mathbb{R}_+)$ denotes the Lebesgue space of measurable functions $h$ on $\mathbb{R}_+$ such that
	$$
	\begin{array}{l} 
		\|h\|_{\gamma,p}=\left(\int_{0}^\infty|f(s)|^{p} s^{2\gamma+1} ds\right)^{\frac{1}{p}}<\infty, \quad \text { if } 1 \leq p<\infty, \\
		\|h\|_{\gamma,\infty}=\displaystyle\supess_{s\in\mathbb{R}_+}|f(s)|<\infty, \quad \text { if } p=\infty,
	\end{array}
	$$
	provided with the topology defined by the norm $\|\cdot\|_{\gamma ,p}$, with $ \gamma > -1/2$.
	\item \(\mathcal{S}_*(\mathbb{R})\) denote the Schwartz space consisting of even, infinitely differentiable (\(\mathcal{C}^{\infty}\)) functions on \(\mathbb{R}\) that, along with all their derivatives, decay rapidly at infinity, equipped with the topology generated by the family of seminorms:
	
	\[
	q_{n, m}(h) = \sup_{s \geq 0} \left(1 + s^2\right)^n \left|\frac{\mathrm{d}^m}{\mathrm{d}s^m} h(s)\right|, \quad h \in \mathcal{S}_*(\mathbb{R}), \quad n, m \in \mathbb{N}.
	\]

\end{itemize} 

\begin{lemma} Let  \( a, b, c, d, e \)  five real parameters such that $b\neq 0$ and $ \gamma > -1/2$.
 Then for all  $v,w\in \mathbb{R}_+$, we have:	
		\begin{equation}\label{ModuleQFDTK}
			\left|\mathbb{J}^{a,b,c}_{d,e,\gamma}(t,s)\right|\leq 1.
		\end{equation}
\end{lemma}

\begin{proof} 
For all \( a, b, c, d, e \in \mathbb{R}\) such that $b\neq 0$ and $ \gamma > -1/2$, we have 
\begin{align*}
  \left|\mathbb{J}^{a,b,c}_{d,e,\gamma}(t,s)\right|&= \left| e^{-i(as^2+ct^2+ds+et)} j_{\gamma}\left(st/b\right) \right|\\
  &\leq  \left| j_{\gamma}\left(st/b\right) \right| .
\end{align*}
According to \cite[Corollary 5.4]{rosler1999positivity}, we deduce that

\begin{equation}
	\left|\mathbb{J}^{a,b,c}_{d,e,\gamma}(t,s)\right|\leq 1.
\end{equation}

\end{proof}

\begin{theorem} 
Let \( a, b, c, d, e \in \mathbb{R}\) such that $b\neq 0$,  $ \gamma > -1/2$. 
\begin{enumerate}
 \item If  $h$ be a function belongs to  $L^{1}_{\gamma}(\mathbb{R}_+)$. Then the maps $t\mapsto\mathbb{B}^{a,b,c}_{d,e,\gamma}[h](t)$ is continuous on $\mathbb{R}_+$ and we have
 
 \begin{equation}
 	\left|\mathbb{B}^{a,b,c}_{d,e,\gamma}[h](t)\right|\leq \frac{c_\gamma}{ |b|^{\gamma+1}}	\|h\|_{\gamma,1}.
 \end{equation}
  \item For all $h\in \mathcal{S}_*(\mathbb{R})$, we have
\begin{equation}\label{QFFBTvsBT}
\mathbb{B}^{a,b,c}_{d,e,\gamma}[h](t) =  \frac{e^{-i(ct^2 + et)}}{(i b)^{\gamma+1}}  \mathcal{B}_{\gamma}\left[e^{-i(as^2 + ds)} h(s)\right]\left(\frac{t}{b}\right).
\end{equation}

\end{enumerate}
\end{theorem}

\begin{proof} We consider  \( a, b, c, d, e \in \mathbb{R}\) such that $b\neq 0$,  $ \gamma > -1/2$.
\begin{enumerate}
	\item Let $h$ be function in  $L^{1}_{\gamma}(\mathbb{R}_+)$. Since the maps $s\mapsto j_{\gamma}\left(st/b\right)$ is continuous on $\mathbb{R}_+$, we deduce that
	 $$w\mapsto\mathbb{B}^{a,b,c}_{d,e,\gamma}[h](s),$$ 
	  is also continuous maps on the half  line for every \( s \in \mathbb{R}_+ \). Furthermore,
	\begin{eqnarray*}
         	\left|\mathbb{B}^{a,b,c}_{d,e,\gamma}[h](s)\right| &=& \left| \frac{c_\gamma}{( i b)^{\gamma+1}} \int_{0}^\infty \mathbb{J}^{a,b,c}_{d,e,\gamma}(t,s) f(s)s^{2\gamma+1} ds \right| \\
         	&\leq&  \frac{c_\gamma}{ |b|^{\gamma+1}} \int_{0}^\infty \left|\mathbb{J}^{a,b,c}_{d,e,\gamma}(t,s) f(s)\right| s^{2\gamma+1} ds  \\
         		&\leq&  \frac{c_\gamma}{ |b|^{\gamma+1}} \int_{0}^\infty \left| f(s)\right| s^{2\gamma+1} ds  \\
         		& \leq& \frac{c_\gamma}{ |b|^{\gamma+1}}	\|h\|_{\gamma,1}.
	\end{eqnarray*}
	
	\item Let   $h$ be a function in $ \mathcal{S}\left(\mathbb{R}_+\right)$, then we have
	\begin{eqnarray*}
		\mathbb{B}^{a,b,c}_{d,e,\gamma}[h](s)&=& \frac{c_\gamma}{( i b)^{\gamma+1}} \int_{0}^\infty \mathbb{J}^{a,b,c}_{d,e,\gamma}(t,s) f(s)s^{2\gamma+1} ds
		\\
		&=& \frac{c_\gamma}{(i b)^{\gamma+1}} \int_{0}^\infty e^{-i(as^2 + ct^2 + ds + et)} j_{\gamma}\left(\frac{st}{b}\right) h(s) s^{2\gamma+1} ds
		\\
		&=& \frac{c_\gamma}{(i b)^{\gamma+1}} e^{-i(ct^2 + et)} \int_{0}^\infty e^{-i(as^2 + ds)} j_{\gamma}\left(\frac{st}{b}\right) h(s) s^{2\gamma+1} ds
		\\
		&=& \frac{e^{-i(ct^2 + et)}}{(i b)^{\gamma+1}}  \left(c_\gamma\int_{0}^\infty j_{\gamma}\left(\frac{st}{b}\right) \left(  e^{-i(as^2 + ds)}  h(s)\right)  s^{2\gamma+1} ds\right)
		\\
		&=& \frac{e^{-i(ct^2 + et)}}{(i b)^{\gamma+1}}  \mathcal{B}_{\gamma}\left[e^{-i(as^2 + ds)} h(s)\right]\left(\frac{t}{b}\right).
	\end{eqnarray*}
\end{enumerate}  
\end{proof}

We deduce immediately, from the above Theorem, the Riemann–Lebesgue lemma for the  quadratic-phase Fourier--Bessel  transform.
\begin{corollary}{(Riemann–Lebesgue lemma)}
Let  \( a, b, c, d, e \)  five real parameters such that $b\neq 0$ and $ \gamma > -1/2$. For every $h\in L^{1}_{\gamma}(\mathbb{R}_+)$, the quadratic-phase Fourier--Bessel  transform belongs $\mathbb{B}^{a,b,c}_{d,e,\gamma}[h]$ to  \( \mathcal{C}_{*, 0}(\mathbb{R}) \) such that

\begin{equation}\label{RLLQFFT}
	\left\|\mathbb{B}^{a,b,c}_{d,e,\gamma}[h](s)\right\|_{\gamma,\infty}\leq \frac{c_\gamma}{ |b|^{\gamma+1}}	\|h\|_{\gamma,1}.
\end{equation}

\end{corollary}

\begin{theorem} 
	Let \( a, b, c, d, e \in \mathbb{R}\) such that $b\neq 0$,  $ \gamma > -1/2$, $\alpha,\beta\in\mathbb{R}$ and $k\in\mathbb{R}_+$. 
  For every $h,g\in L^{1}_{\gamma}(\mathbb{R}_+)$, the quadratic-phase Fourier--Bessel  transform satisfies the following properties:
	\begin{enumerate}
		\item Linearity: $\mathbb{B}^{a,b,c}_{d,e,\gamma}[\alpha h+\beta g](s) = \alpha\mathbb{B}^{a,b,c}_{d,e,\gamma}[h](s) + \beta\mathbb{B}^{a,b,c}_{d,e,\gamma}[g](s)$.

		\item Scaling: $ \mathbb{B}^{a,b,c}_{d,e,\gamma}[h](k t)=  \frac{1}{k^{2\gamma+2}} \mathbb{B}^{a',b,c'}_{d',e',\gamma} [h_k](t) $, with   $ a' = \frac{a}{k^2}$, $c' = c k^2$, $d' = \frac{d}{k}$,	$e' = e k$ and 
		$ h_k(v) = h(v/k)$.
\end{enumerate}
\end{theorem}

\begin{proof} Let $h,g\in L^{1}_{\gamma}(\mathbb{R}_+)$,  \( a, b, c, d, e \in \mathbb{R}\) such that $b\neq 0$,  $ \gamma > -1/2$, $k,\beta\in\mathbb{R}$ and  $k\in\mathbb{R}_+$.
   	\begin{enumerate}
   	\item  Linearity: 
   		\begin{eqnarray*}
   		\mathbb{B}^{a,b,c}_{d,e,\gamma}[\alpha h+\beta g](s) &=& \frac{c_\gamma}{( i b)^{\gamma+1}} \int_{0}^\infty \mathbb{J}^{a,b,c}_{d,e,\gamma}(t,s) [\alpha h+\beta g](s)s^{2\gamma+1} ds
   		\\
   		 &=& \frac{ \alpha c_\gamma}{( i b)^{\gamma+1}} \int_{0}^\infty \mathbb{J}^{a,b,c}_{d,e,\gamma}(t,s) h(s)s^{2\gamma+1} ds\\
   	  && + \frac{ \beta c_\gamma}{( i b)^{\gamma+1}} \int_{0}^\infty \mathbb{J}^{a,b,c}_{d,e,\gamma}(t,s) g(s)s^{2\gamma+1} ds\\
   	   &=& \alpha\mathbb{B}^{a,b,c}_{d,e,\gamma}[h](s) + \beta\mathbb{B}^{a,b,c}_{d,e,\gamma}[g](s).
   	\end{eqnarray*}
 \item Scaling: we have  	
   	\begin{eqnarray*}
   		\mathbb{B}^{a,b,c}_{d,e,\gamma}[h](k t) &=& \frac{c_\gamma}{( i b)^{\gamma+1}} \int_{0}^\infty \mathbb{J}^{a,b,c}_{d,e,\gamma}(k t,s) h(s)s^{2 \gamma+1} d s
   		\\
   		&=& \frac{c_\gamma}{( i b)^{\gamma+1}} \int_{0}^\infty e^{-i(as^2+c(k t)^2+ds+e(k t))} j_{\gamma}\left(\frac{s(k t)}{b}\right) h(s)s^{2 \gamma+1} d s.\\
   \end{eqnarray*}	
   We make the change of variable $s = \frac{u}{k}$ to simplify the expression, we get
   \begin{eqnarray*}
   \mathbb{B}^{a,b,c}_{d,e,\gamma}[h](k t) &=& \frac{c_\gamma}{( i b)^{\gamma+1}} \int_{0}^\infty e^{-i \left( \frac{a}{k^2} u^2 + c k^2 t^2 + \frac{d}{k} u + e k t \right)} j_{\gamma} \left(\frac{u t}{b}\right) h\left(\frac{u}{k}\right) u^{2\gamma+1} du
   \\
   &=& \frac{1}{k^{2\gamma+2}}  \frac{c_\gamma}{( i b)^{\gamma+1}} \int_{0}^\infty \mathbb{J}^{a',b,c'}_{d',e',\gamma}(u,s) h_k(u) u^{2 \gamma+1} d u
   \\
   &=& \frac{1}{k^{2\gamma+2}} \mathbb{B}^{a',b,c'}_{d',e',\gamma} [h_k](t) ,
\end{eqnarray*}	
 with   $ a = \frac{a}{k^2}$, $c = c k^2$, $d = \frac{d}{k}$,	$e = e k$ and 
$ h_k(v) = h(v/k)$.
   \end{enumerate}

\end{proof}

\begin{theorem}{(The reversibility property)}\label{invthmQF.FBT}  Let   \( a, b, c, d, e \in \mathbb{R}\) such that $b\neq 0$, and $ \gamma > -1/2$. Let $h$ be a function belonging to $\in L^{1}_{\gamma}(\mathbb{R}_+)$, and  $\mathbb{B}^{a,b,c}_{d,e,\gamma}[h]$ belongs to $L^{1}_{\gamma}(\mathbb{R}_+)$. Then the inverse of  the quadratic-phase Fourier--Bessel  transform is given by

\begin{eqnarray}\label{InvQF.FBT} 
	    h(s) &=& \left( \mathbb{B}^{-c,-b,-a}_{-e,-d,\gamma} \left( \mathbb{B}^{a,b,c}_{d,e,\gamma}[h] \right)(t) \right)(s) \nonumber
	\\
	&=& \frac{c_\gamma}{( i b)^{\gamma+1}} \int_{0}^\infty \mathbb{J}^{-c,-b,-a}_{-e,-d,\gamma}(s,t) \mathbb{B}^{a,b,c}_{d,e,\gamma}[h](t) t^{2 \gamma+1} dt.
\end{eqnarray}
\end{theorem}

\begin{proof} 

 Let $h$ be a function belonging to $\in L^{1}_{\gamma}(\mathbb{R}_+)$, we have
 
 \begin{eqnarray*}
 \left( \mathbb{B}^{-c,-b,-a}_{-e,-d,\gamma} \left( \mathbb{B}^{a,b,c}_{d,e,\gamma}[h] \right)(t) \right)(s) 
 	&=& \frac{c_\gamma}{( -i b)^{\gamma+1}} \int_{0}^\infty \mathbb{J}^{-c,-b,-a}_{-e,-d,\gamma}(s,t) \mathbb{B}^{a,b,c}_{d,e,\gamma}[h](t) t^{2 \gamma+1} dt
 	\\
 	&=& \frac{c_\gamma}{( -i b)^{\gamma+1}}  \int_{0}^\infty \mathbb{J}^{-c,-b,-a}_{-e,-d,\gamma}(s,t)   \bigg[ \frac{c_\gamma}{( i b)^{\gamma+1}} 
 	\\
 	&&\quad\times\int_{0}^\infty \mathbb{J}^{a,b,c}_{d,e,\gamma}(t,u) f(u)u^{2\gamma+1} du \bigg] t^{2 \gamma+1} dt.
  \end{eqnarray*}	
 	Using Fubini’s theorem, we may interchange the order of integration and obtain
 
   \begin{eqnarray*}
 	 \left( \mathbb{B}^{-c,-b,-a}_{-e,-d,\gamma} \left( \mathbb{B}^{a,b,c}_{d,e,\gamma}[h] \right)(t) \right)(s)
 	&=& \frac{c_\gamma^2}{( i b)^{\gamma+1}( -i b)^{\gamma+1}} \int_0^\infty \int_0^\infty  \mathbb{J}^{-c,-b,-a}_{-e,-d,\gamma}(s,t) \mathbb{J}^{a,b,c}_{d,e,\gamma}(t,u) 
 	\\ \\
 	&& \quad\times h(u) u^{2\gamma+1} t^{2\gamma+1} dt du.
 \end{eqnarray*}
By evaluating the product of the kernels, we obtain
 
   \begin{eqnarray*}
 \mathbb{J}^{-c,-b,-a}_{-e,-d,\gamma}(s,t)	\mathbb{J}^{a,b,c}_{d,e,\gamma}(t,u) 
  &=&  e^{i(ct^2+as^2+ds+et)} j_{\gamma}\left(-st/b\right)  e^{-i(au^2+ct^2+du+et)} j_{\gamma}\left(ut/b\right).
\end{eqnarray*}
 Using the fact that the normalized Bessel function $j_{\gamma}$ is even, we get
 
   \begin{eqnarray*}
	\mathbb{J}^{-c,-b,-a}_{-e,-d,\gamma}(s,t)	\mathbb{J}^{a,b,c}_{d,e,\gamma}(t,u) 
 &=&   e^{i(as^2+ds)} e^{-i(au^2+du)} j_{\gamma}\left(st/b\right) j_{\gamma}\left(ut/b\right).
 \end{eqnarray*}
 Hence, we obtain
 
  \begin{eqnarray*}
 \left( \mathbb{B}^{-c,-b,-a}_{-e,-d,\gamma} \left( \mathbb{B}^{a,b,c}_{d,e,\gamma}[h] \right)(t) \right)(s)
 &=& \frac{c_\gamma^2}{( i b)^{\gamma+1}( -i b)^{\gamma+1}}  e^{i(as^2+ds)} \int_0^\infty e^{-i(a u^2 + d u)} h(u)  \\
 &&\quad\times \left[  \int_0^\infty j_{\gamma}\left(st/b\right) j_{\gamma}\left(ut/b\right)  t^{2\gamma+1} dt  \right] u^{2\gamma+1} du.
 \end{eqnarray*}
By performing the change of variables $v = t/b$, we obtain

\begin{eqnarray*}
	\left( \mathbb{B}^{-c,-b,-a}_{-e,-d,\gamma} \left( \mathbb{B}^{a,b,c}_{d,e,\gamma}[h] \right)(t) \right)(s)
	&=& c_\gamma^{2} e^{i(as^2+ds)} \int_0^\infty e^{-i(a u^2 + d u)} h(u)  \\
	&& \quad \times \left[  \int_0^\infty j_{\gamma}(s v) j_{\gamma}(u v) v^{2\gamma+1} dv  \right]
	 u^{2\gamma+1} du
	 \\
	 &=& c_\gamma^{2} e^{i(as^2+ds)} \int_0^\infty  j_{\gamma}(s v) v^{2\gamma+1}      
	  \\
	 && \quad \times \left[  \int_0^\infty  j_{\gamma}(u v)  e^{-i(a u^2 + d u)} h(u) u^{2\gamma+1} du  \right]
	 dv.
	 \\  &=& c_\gamma e^{i(as^2+ds)} \int_0^\infty  j_{\gamma}(s v) 	\mathcal{B}_{\gamma} [H](v) v^{2\gamma+1} dv
 \end{eqnarray*}
where $H(u)= e^{-i(a u^2 + d u)} h(u)$.

Consequently, we obtain

\begin{eqnarray*}
	\left( \mathbb{B}^{-c,-b,-a}_{-e,-d,\gamma} \left( \mathbb{B}^{a,b,c}_{d,e,\gamma}[h] \right)(t) \right)(s)
	&=&   e^{i(as^2+ds)} \left[	\mathcal{B}_{\gamma} \left( \mathcal{B}_{\gamma} [H](s)\right) \right] \\
		&=&   e^{i(as^2+ds)} \left[  H(s) \right] \\
			&=&   e^{i(as^2+ds)} e^{-i(a s^2 + d s)} h(s)\\
			&=& h(s).
 \end{eqnarray*}

Therefore,
\[
\mathbb{B}^{-c,-b,-a}_{-e,-d,\gamma} \left( \mathbb{B}^{a,b,c}_{d,e,\gamma}[h] \right) (s) = h(s).
\]

This confirms that the quadratic-phase Fourier--Bessel transform is invertible, with the inverse given by the relation

\[
\mathbb{B}^{-c,-b,-a}_{-e,-d,\gamma} \circ \mathbb{B}^{a,b,c}_{d,e,\gamma} = id.
\]

\end{proof}   

Now we present and demonstrate the Parseval identity for the quadratic phase Fourier--Bessel transform.

\begin{theorem}[Parseval’s formula]
	Let \( a, b, c, d, e \in \mathbb{R}\),  $\gamma \geq -1/2$ and  $h$ and $g$ be two functions belonging to $L^{2}_{\gamma}(\mathbb{R}_+)$, then we have
	
	\begin{equation}\label{PFQFDT}
		\int_{\mathbb{R}_+} h(s) \overline{g(s)} |s|^{2\gamma+1} ds = \int_{\mathbb{R}_+} \mathbb{B}^{a,b,c}_{d,e,\gamma}[h](t) \overline{\mathbb{B}^{a,b,c}_{d,e,\gamma}[g](t)} t^{2\gamma+1} d t.
	\end{equation}
\end{theorem}

\begin{proof}
	Let  $h$ and $g$ be two functions in $L^{2}_{\gamma}(\mathbb{R}_+)$. Using the relation (\ref{QFFBTvsBT}) between the Fourier--Bessel transform and its quadratic phase version, we get
	
	\begin{equation*}
		\mathbb{B}^{a,b,c}_{d,e,\gamma}[h](t) = \frac{1}{( i b)^{\gamma+1}} e^{-i(ct^2+et)} \mathcal{B}_\gamma(h_1)(t / b),
	\end{equation*}
	where $h_1(s)=e^{-i(as^2+ds)}h(s),$ and 
	
	\begin{equation*}
		\overline{\mathbb{B}^{a,b,c}_{d,e,\gamma}[g](t)} = \frac{1}{( -i b)^{\gamma+1}} e^{i(ct^2+et)} \overline{\mathcal{B}_\gamma(g_1)(t / b)},
	\end{equation*}
	with $g_1(s)=e^{-i(as^2+ds)}g(s).$\\
	Subsequently, we have
	
	\begin{eqnarray*}
		\int_{\mathbb{R}_+} \mathbb{B}^{a,b,c}_{d,e,\gamma}[h](t) \overline{\mathbb{B}^{a,b,c}_{d,e,\gamma}[g](t)} t^{2\gamma+1} d t
		&=& \frac{1}{( i b)^{\gamma+1}}\frac{1}{( -i b)^{\gamma+1}}\int_{\mathbb{R}_+} \mathcal{B}_\gamma(h_1)(t / b) 
		\\ \\
		&& \times \overline{\mathcal{B}_\gamma(g_1)(t / b)} t^{2\gamma+1} d t
		\\
		&=& \frac{1}{|b|^{2\gamma+2}} \int_{\mathbb{R}_+} \mathcal{B}_\gamma(h_1)(t / b) \overline{\mathcal{B}_\gamma(g_1)(t / b)} t^{2\gamma+1} d t.
	\end{eqnarray*}
	Applying the change of variable $(z = t/ b)$, we get
	
	\begin{eqnarray*}
		\int_{\mathbb{R}_+} \mathbb{B}^{a,b,c}_{d,e,\gamma}[h](t) \overline{\mathbb{B}^{a,b,c}_{d,e,\gamma}[g](t)} t^{2\gamma+1} d t
		&=&  \int_{\mathbb{R}_+} \mathcal{B}_\gamma(h_1)(z) \overline{\mathcal{B}_\gamma(g_1)(z)} z^{2\gamma+1} d z.
	\end{eqnarray*}
	By the Plancherel theorem for the Fourier--Bessel transform, we get
	
	\begin{eqnarray*}
		\int_{\mathbb{R}_+} \mathbb{B}^{a,b,c}_{d,e,\gamma}[h](t) \overline{\mathbb{B}^{a,b,c}_{d,e,\gamma}[g](t)} t^{2\gamma+1} d t
		&=&  \int_{\mathbb{R}_+} h_1(z) \overline{g_1(z)} z^{2\gamma+1} d z
		\\
		&=&  \int_{\mathbb{R}_+} e^{-i(az^2+dz)}h(z) \overline{e^{-i(az^2+dz)}g(z)} z^{2\gamma+1} d z
		\\
		&=&  \int_{\mathbb{R}_+} h(z) \overline{g(z)} z^{2\gamma+1} d z.		
	\end{eqnarray*}
	This completes the verification of the identity, thereby concluding the proof.
\end{proof}

\begin{corollary}\label{PlanchQPDT_S}
	Let \( a, b, c, d, e \in \mathbb{R}$ with $b\neq 0$,  $\gamma \geq -1/2$, and $h$ be a function belonging to $\mathcal{S}_*(\mathbb{R})$. Then,
	
	\begin{equation*}
		\|\mathbb{B}^{a,b,c}_{d,e,\gamma}[h]\|_{\gamma,2}=  \|h\|_{\gamma,2}.
	\end{equation*}
\end{corollary}

\begin{theorem}\label{PlanchQPDT}
	Let \( a, b, c, d, e \in \mathbb{R}$ such that $b\neq 0$,  $\gamma \geq -1/2$.
	If $h \in L^{1}_{\gamma}(\mathbb{R}_+) \cap L^{2}_{\gamma}(\mathbb{R}_+)$, then its quadratic-phase Fourier--Bessel transform $\mathbb{B}^{a,b,c}_{d,e,\gamma}[h]$ belongs to $L^{2}_{\gamma}(\mathbb{R}_+)$, and we have
	\begin{equation}\label{PlanchQPDT_L2}
		\|\mathbb{B}^{a,b,c}_{d,e,\gamma}[h]\|_{\gamma,2}=  \|h\|_{\gamma,2}.
	\end{equation}
\end{theorem}

\begin{proof} 
	Let \( a, b, c, d, e \in \mathbb{R}$ such that $b\neq 0$,  $\gamma \geq -1/2$.
	By Corollary \ref{PlanchQPDT_S} and the density of $\mathcal{S}_*(\mathbb{R})$ in $L^{2}_{\gamma}(\mathbb{R}_+)$, we obtain a unique continuous operator $\mathbf{B}^{a,b,c}_{d,e,\gamma}[h]$ on $L^{2}_{\gamma}(\mathbb{R}_+)$, which coincides with $\mathbb{B}^{a,b,c}_{d,e,\gamma}[h]$ on  $\mathcal{S}_*(\mathbb{R})$. Let $h$ and $g$ be two functions in $\mathcal{S}_*(\mathbb{R})$; then we have 
	
	\begin{eqnarray*}
		\int_{\mathbb{R}_+} \mathbf{B}^{a,b,c}_{d,e,\gamma}[h](t) \overline{\mathbb{B}^{a,b,c}_{d,e,\gamma}[g](t)} t^{2\gamma+1} d t
		&=& \int_{\mathbb{R}_+} \mathbb{B}^{a,b,c}_{d,e,\gamma}[h](t) \overline{\mathbb{B}^{a,b,c}_{d,e,\gamma}[g](t)} t^{2\gamma+1} d t  \\
		&=& \int_{\mathbb{R}_+} \mathbb{B}^{a,b,c}_{d,e,\gamma}[h](t) \overline{\mathbf{B}^{a,b,c}_{d,e,\gamma}[g](t)} t^{2\gamma+1} d t.
	\end{eqnarray*}
	Consequently, by the density of $\mathcal{S}_*(\mathbb{R})$ in $L^{2}_{\gamma}(\mathbb{R}_+)$, it follows that for all $h$ and $g$ in $L^{2}_{\gamma}(\mathbb{R}_+)$, 
	
	\begin{eqnarray*}
		\int_{\mathbb{R}_+} \mathbf{B}^{a,b,c}_{d,e,\gamma}[h](t) \overline{\mathbb{B}^{a,b,c}_{d,e,\gamma}[g](t)} t^{2\gamma+1} d t
		&=& \int_{\mathbb{R}_+} \mathbb{B}^{a,b,c}_{d,e,\gamma}[h](t) \overline{\mathbf{B}^{a,b,c}_{d,e,\gamma}[g](t)} t^{2\gamma+1} d t.
	\end{eqnarray*}
	
	Now, assume that $h$ is an element of $ L^{1}_{\gamma}(\mathbb{R}_+) \cap L^{2}_{\gamma}(\mathbb{R}_+)$ and $g$ belongs to $\mathcal{S}_*(\mathbb{R})$. Then, we have
	
	\begin{eqnarray*}
		\int_{\mathbb{R}_+} \mathbb{B}^{a,b,c}_{d,e,\gamma}[h](t) \overline{\mathbb{B}^{a,b,c}_{d,e,\gamma}[g](t)} t^{2\gamma+1} d t
		&=& \int_{\mathbb{R}_+} \mathbb{B}^{a,b,c}_{d,e,\gamma}[h](t) \overline{\mathbf{B}^{a,b,c}_{d,e,\gamma}[g](t)} t^{2\gamma+1} d t  \\
		&=& \int_{\mathbb{R}_+} \mathbf{B}^{a,b,c}_{d,e,\gamma}[h](t) \overline{\mathbb{B}^{a,b,c}_{d,e,\gamma}[g](t)} t^{2\gamma+1} d t.
	\end{eqnarray*}
	Therefore, $\mathbb{B}^{a,b,c}_{d,e,\gamma}[h]=\mathbf{B}^{a,b,c}_{d,e,\gamma}[h]$ almost everywhere. This shows that for every $h \in L^{1}_{\gamma}(\mathbb{R}_+) \cap L^{2}_{\gamma}(\mathbb{R}_+)$, its quadratic-phase Fourier--Bessel transform $\mathbb{B}^{a,b,c}_{d,e,\gamma}[h]$ is in $L^{2}_{\gamma}(\mathbb{R}_+)$. The equality (\ref{PlanchQPDT_L2}) follows directly from  Corollary \ref{PlanchQPDT_S}.
\end{proof}

\section{Quadratic-phase Fourier--Bessel  translation} \label{sec3}

\begin{definition}
	Let   \( a, b, c, d, e \in \mathbb{R}\) such that $b\neq 0$, and $ \gamma > -1/2$. . For $h \in \mathcal{C}_{*}(\mathbb{R}_+)$, we define the generalized translation operators associated with the quadratic-phase Fourier--Bessel transform by
	
	\begin{equation}\label{QFFBTtransl}
		  \mathbb{T} ^{a,b,d}_{t,\gamma}[h](s)= \int_0^\infty h(u) \mathbb{W}^{a,b,d}_{\gamma}(s,t,u) e^{au^2+du}u^{2\gamma+1} du,
	\end{equation}
	where the kernel $\mathbb{W}^{a,b,d}_{\gamma}(s,t,u)$ is given by

%
%

	\begin{equation*}
		\mathbb{W}^{a,b,d}_{\gamma}(s,t,u)= \mathbb{W}_{\gamma}(s,t,u)e^{-i\left[ a(s^2+t^2+u^2) +d(s+t+u)\right]},
	\end{equation*}
	with  $\mathbb{W}_{\gamma}(s,t,u)$ is given by

\begin{equation*}\label{QFFBkernelWclassicTransl1}
	\mathbb{W}_{\gamma}(s,t,u)=\frac{2^{2\gamma-1} \Gamma(\gamma+1)}{\sqrt{\pi} \Gamma\left(\gamma+\frac{1}{2}\right)} \frac{\Delta^{2\gamma-1}}{(stu)^{2\gamma}} \chi_{[|s-t|, s+t]},
\end{equation*}
and represent the kernel of the translation operator associated with the Fourier--Bessel transform, given by \cite{levitan1951expansion}

	\begin{equation*}
		\mathcal{T}_{t,\gamma}[h](s)= \int_0^\infty h(u) \mathbb{W}_{\gamma}(s,t,u) u^{2\gamma+1} du,
	\end{equation*}
and
	
	$$
	\Delta=\frac{1}{4} \sqrt{(x+y+z)(x+y-z)(x-y+z)(y+z-x)},
	$$
	represents the area of a triangle with side lengths \( s, t, u > 0 \), and \( \chi_A \) denotes the characteristic function of the set \( A \).
	\end{definition}

 \begin{properties}
	Let   \( a, b, c, d, e \in \mathbb{R}\) such that $b\neq 0$, and $ \gamma > -1/2$. . For all $h,g \in \mathcal{C}_{*}(\mathbb{R}_+)$ and
	$\alpha,\beta\in\mathbb{R}$, the quadratic-phase Fourier--Bessel translation operator satisfies the following properties.
	\begin{enumerate}
		\item Identity: $  \mathbb{T} ^{a,b,d}_{0,\gamma}=Id$.
	    \item Symmetry: $  \mathbb{T} ^{a,b,d}_{t,\gamma}[h](s)=  \mathbb{T} ^{a,b,d}_{s,\gamma}[h](t)$.
	    \item Linearity: $  \mathbb{T} ^{a,b,d}_{t,\gamma}[\alpha h+\beta g](s)= \alpha   \mathbb{T} ^{a,b,d}_{t,\gamma}[h](s) + \beta  \mathbb{T} ^{a,b,d}_{t,\gamma}[g](s)$.
	    \item Compact support: If $h(s)=0$ for $y\geq r$, then $  \mathbb{T} ^{a,b,d}_{t,\gamma}[h](s)=0$, for all $|x-y| \geq r$.
	\end{enumerate}
	 \end{properties}

In the following proposition, we prove the continuity of the quadratic-phase Fourier--Bessel translation operator from  $L^{p}_{\gamma}(\mathbb{R}_+)$  into itself for all $p\in [1,\infty ]$.

\begin{proposition} For all function $h$ belongs to $L^{p}_{\gamma}(\mathbb{R}_+)$, $p\in [1,\infty ]$ and $s\in [0,\infty)$, the function $\mathbb{T} ^{a,b,d}_{t,\gamma}[h]$ belongs to $L^{p}_{\gamma}(\mathbb{R}_+)$ such that
	
\begin{equation}\label{contQFFBTransl}
 \|\mathbb{T} ^{a,b,d}_{t,\gamma}[h]\|_{\gamma,p} \leq \|h\|_{\gamma,p}.
\end{equation}
\end{proposition}
\begin{proof} For $p=1$ and $\infty$, the proof is trivial. Let us now consider $p\in (1,\infty)$ and $h$ belongs to $L^{p}_{\gamma}(\mathbb{R}_+)$. Then from the definition of the quadratic-phase Fourier--Bessel translation operator (\ref{QFFBTtransl}), we have
	
 \begin{eqnarray*}  \left|  \mathbb{T} ^{a,b,d}_{t,\gamma}[h](s)\right|  &=& \left| \int_0^\infty h(u) \mathbb{W}^{a,b,d}_{\gamma}(s,t,u) e^{au^2+du}u^{2\gamma+1} du\right| \\
 	&\leq&  \int_0^\infty |h(u)| \left|\mathbb{W}^{a,b,d}_{\gamma}(s,t,u)\right| u^{2\gamma+1} du,\\
 		&=&  \int_0^\infty |h(u)| \left|\mathbb{W}_{\gamma}(s,t,u)e^{-i\left[ a(s^2+t^2+u^2) +d(s+t+u)\right]}\right| u^{2\gamma+1} du\\
 		&=&  \int_0^\infty |h(u)| \left|\mathbb{W}_{\gamma}(s,t,u)\right| u^{2\gamma+1} du.
  \end{eqnarray*}	
 According to Hölder's inequality, and the fact that 
 \begin{equation*}
   \int_0^\infty \mathbb{W}_{\gamma}(s,t,u)\ u^{2\gamma+1} du=1,
 \end{equation*}
  we obtain that
  \begin{equation*}
  	\|\mathbb{T} ^{a,b,d}_{t,\gamma}[h]\|_{\gamma,p} \leq \|h\|_{\gamma,p}.
  \end{equation*}
  
\end{proof}

\section{Quadratic-phase Fourier--Bessel  convolution}\label{sec4}

\begin{definition}
	Let   \( a, b, c, d, e \in \mathbb{R}\) such that $b\neq 0$, and $ \gamma > -1/2$. For \(h, g \in L^1_{\gamma}(\mathbb{R}_+)\), we define the generalized convolution product associated with the quadratic-phase Fourier--Bessel transform by
\end{definition}

\begin{equation}\label{QFFBTconv}
h\star g(t)=  \int_0^\infty \mathbb{T} ^{a,b,d}_{t,\gamma}[h](s) g(s) e^{i(as^2+ds)}s^{2\gamma+1} ds,
\end{equation}
where $\mathbb{T} ^{a,b,d}_{t,\gamma}[h]$ is the generalized translation operators associated with the quadratic-phase Fourier--Bessel transform  given by equation (\ref{QFFBTtransl}).
 
The basic properties of the quadratic-phase Fourier--Bessel convolution product are collected in the following statement.

\begin{properties} Let   \( a, b, c, d, e \in \mathbb{R}\) such that $b\neq 0$, $ \gamma > -1/2$ and \(h, g, \upsilon \in L^1_{\gamma}(\mathbb{R}_+)\), then the quadratic-phase Fourier--Bessel convolution product satisfies the following properties:
	\begin{enumerate}
     \item Commutativity: $h\star g= g\star h$.
     \item  Associativity: $(h\star g)\star \upsilon =  h\star (g\star \upsilon)$.
	\end{enumerate}
\end{properties}

In the following, we state the Young's inequality for the quadratic-phase Fourier--Bessel convolution product.
\begin{proposition}
	Let   \( a, b, c, d, e \in \mathbb{R}\) such that $b\neq 0$, and $ \gamma > -1/2$. Assume that $1 \leq p, q, r \leq \infty$ and $\frac{1}{p}+\frac{1}{q}=\frac{1}{r}+1$. If \(h \in L^p_{\gamma}(\mathbb{R}_+)\) and \(g \in L^q_{\gamma}(\mathbb{R}_+)\), then $h\star g \in L^r_{\gamma}(\mathbb{R}_+)$ and

$$
\|h\star g\|_{\gamma,r} \leq\|h\|_{\gamma,p}\|g\|_{\gamma,q}.
$$

In particular, if $g, h \in L^1_{\gamma}(\mathbb{R}_+)$, then the  quadratic-phase Fourier--Bessel convolution product $h\star g $ is defined almost everywhere on $\mathbb{R}_+$ and it belongs to $L^1_{\gamma}(\mathbb{R}_+)$.
\end{proposition}

\begin{proof}
Let   \( a, b, c, d, e \in \mathbb{R}\) such that $b\neq 0$, and $ \gamma > -1/2$. We have
  
 \begin{eqnarray*}
  	 \left|\mathbb{T} ^{a,b,d}_{t,\gamma}[h](s) g(s) e^{i(as^2+ds)}\right| 
  	 = \left(\left|\mathbb{T} ^{a,b,d}_{t,\gamma}[h](s) \right|^p|g(s)|^q\right)^{\frac{1}{r}}\left(\left|\mathbb{T} ^{a,b,d}_{t,\gamma}[h](s) \right|^p\right)^{\frac{1}{p}-\frac{1}{r}}\left(|g(s)|^q\right)^{\frac{1}{q}-\frac{1}{q}}.
 \end{eqnarray*}
Therefore,
  
 \begin{eqnarray*}
  	 \int_0^\infty\left|\mathbb{T} ^{a,b,d}_{t,\gamma}[h](s)  e^{-i \frac{d}{b} y^2} g(s)\right| s^{2\gamma+1} ds
  	& \leq &\left(\int_0^\infty\left|\mathbb{T} ^{a,b,d}_{t,\gamma}[h](s) \right|^p|g(s)|^q s^{2\gamma+1} ds\right)^{\frac{1}{r}} \\
  	& &\times  \left(\int_0^\infty\left|\mathbb{T} ^{a,b,d}_{t,\gamma}[h](s) \right|^p s^{2\gamma+1} ds\right)^{\frac{r-p}{r p}} \\
    & &\times \left(\int_0^\infty|g(s)|^q s^{2\gamma+1} ds\right)^{\frac{r-q}{r q}}.
 \end{eqnarray*}
According to the inequality \ref{contQFFBTransl}, we get  
 
 \begin{eqnarray*}
 	\left|h\star g(t)\right|^r & \leq &\left\|\mathbb{T} ^{a,b,d}_{t,\gamma}[h]\right\|_{\gamma,p}^{r-p}\|g\|_{\gamma,q}^{r-q} \int_0^\infty\left|\mathbb{T} ^{a,b,d}_{t,\gamma}[h](s)\right|^p|g(s)|^q s^{2\gamma+1} ds \\
 	& \leq & \|h\|_{\gamma,p}^{r-p}\|g\|_{\gamma,q}^{r-q} \int_0^\infty\left|\mathbb{T} ^{a,b,d}_{t,\gamma}[h](s)\right|^p|g(s)|^q s^{2\gamma+1} ds.
 \end{eqnarray*}
 It follows that,
 
 \begin{eqnarray*}
 	\left\|h\star g\right\|_{\gamma,r}^r &\leq &  \|h\|_{\gamma,p}^{r-p}\|g\|_{\gamma,q}^{r-q} \int_0^\infty \int_0^\infty\left|\mathbb{T} ^{a,b,d}_{t,\gamma}[h](s)\right|^p|g(s)|^q 
 	 s^{2\gamma+1} ds t^{2\gamma+1} dt \\
 	&\leq &  \|h\|_{\gamma,p}^{r-p}\|g\|_{\gamma,q}^{r-q} \int_0^\infty|g(s)|^q \int_0^\infty\left|\mathbb{T} ^{a,b,d}_{s,\gamma}[h](t)\right|^p 
 	  t^{2\gamma+1} dt  s^{2\gamma+1} ds \\
 	&\leq & \|h\|_{\gamma,p}^{r-p}\|g\|_{\gamma,q}^{r-q}\|h\|_{\gamma,p}^p\|g\|_{\gamma,q}^q \\
 	&\leq & \|h\|_{\gamma,p}^r\|g\|_{\gamma,q}^r.
 \end{eqnarray*}
\end{proof}

\section{Donoho-Stark Uncertainty Principle}\label{sec5}

We aim in this subsection to prove a Donoho–Stark-type uncertainty principle for the quadratic-phase Fourier–Bessel transform. To do this, we begin by recalling the following definitions.

Consider two measurable subsets \(\mathcal{M}\) and \(\mathcal{N}\) of \(\mathbb{R}_+\). A function \( h \) is called \(\epsilon_\mathcal{M}\)-time-limited if we have

\[
\|h - \chi_\mathcal{M} h\|_{\gamma,2} \leq \epsilon_\mathcal{M},
\]
where

\[
\chi_\mathcal{M}(s) =
\begin{cases}
	1, & s \in \mathcal{M}, \\
	0, & \text{otherwise}.
\end{cases}
\]
Similarly, a function \(h\) is said to be \(\epsilon_\mathcal{N}\)-quadratic-phase Fourier--Bessel band-limited  if its quadratic-phase Fourier--Bessel transform \(\mathbb{B}^{a,b,c}_{d,e,\gamma}[h]\) is \(\epsilon_\mathcal{N}\)-limited, in other words

\[
\|\mathbb{B}^{a,b,c}_{d,e,\gamma}[h] - \chi_\mathcal{N} \mathbb{B}^{a,b,c}_{d,e,\gamma}[h]\|_{\gamma,2} \leq \epsilon_\mathcal{N}.
\]

Next, we introduce the  time-limiting operator  \(P_\mathcal{M}\) and the  frequency-limiting operator  \(Q_\mathcal{N}\), defined respectively as

\[
P_\mathcal{M} h(s) = \chi_\mathcal{M}(s)h(s),
\]
and

\[
Q_\mathcal{N}(h)(s) = \mathbb{B}^{-c,-b,-a}_{-e,-d,\gamma} \left( \chi_\mathcal{N}(.) \mathbb{B}^{a,b,c}_{d,e,\gamma}[h](.) \right)(s), \quad s \in \mathbb{R}_+.
\]
Note that \(\|P_\mathcal{M}\| = 1\).

We now state the Donoho-Stark uncertainty principle for the quadratic-phase Fourier--Bessel transform.

  \begin{theorem}
Let \( h \in L^2_{\gamma}(\mathbb{R}_+) \) be a function with \( \|h\|_{L^2_{\gamma}} = 1 \). Consider two measurable subsets \( \mathcal{M} \) and \( \mathcal{N} \) of \( \mathbb{R}_+ \). Assume that the following conditions hold:
  \begin{enumerate}
   \item  \(h\) is \(\epsilon_\mathcal{M}\)-time-limited,
    \item  \(h\) is \(\epsilon_\mathcal{N}\)-quadratic-phase Fourier--Bessel band-limited,
  \end{enumerate}
then, we have

\[ 
|\mathcal{M}||\mathcal{N}| \geq \frac{|b|^{2(\mu + 1)}}{c_\gamma^2} (1 - \epsilon_\mathcal{M} - \epsilon_\mathcal{N})^2,
\]
where \(c_\gamma\) is given by (\ref{constc_gamma}).
\end{theorem}

\begin{proof}

We begin by expressing the unit 2-norm function \(h\) as

\[
h = P_\mathcal{M} h + P_{\mathcal{M}^c} h = P_\mathcal{M} Q_\mathcal{N} h + P_\mathcal{M} Q_{\mathcal{N}^c} h + P_{\mathcal{M}^c} h.
\]
According to  the triangle inequality, we get

\begin{eqnarray}\label{QFFBTeq5.1}
\|h - P_\mathcal{M} Q_\mathcal{N} h\|_{\gamma,2} &=& \|P_\mathcal{M} Q_{\mathcal{N}^c} h + P_{\mathcal{M}^c} h\|_{\gamma,2} 
\nonumber\\
&\leq&\|P_\mathcal{M} Q_{\mathcal{N}^c} h\|_{\gamma,2} + \|P_{\mathcal{M}^c} h\|_{\gamma,2} 
\nonumber\\
&\leq& \epsilon_\mathcal{N} + \epsilon_\mathcal{M}.
\end{eqnarray}
Applying the triangle inequality again, we get

\begin{eqnarray}\label{QFFBTeq5.2}
\|h - P_\mathcal{M} Q_\mathcal{N} h\|_{\gamma,2} \geq \|h\|_{\gamma,2} - \|P_\mathcal{M} Q_\mathcal{N} h\|_{\gamma,2}.
\end{eqnarray}
By merging inequalities (\ref{QFFBTeq5.1}) and (\ref{QFFBTeq5.2}), we get

\[
1 - \|P_\mathcal{M} Q_\mathcal{N} h\|_{\gamma,2} \leq \epsilon_\mathcal{M} + \epsilon_\mathcal{N},
\]
that is,

\begin{eqnarray}\label{QFFBTeq5.3}
\|P_\mathcal{M} Q_\mathcal{N} h\|_{\gamma,2} \geq 1 - \epsilon_\mathcal{M} - \epsilon_\mathcal{N}.
\end{eqnarray}
Alternatively, we must establish an upper limit for the Hilbert-Schmidt norm of \(P_\mathcal{M} Q_\mathcal{N}\).

\[
\|P_\mathcal{M} Q_\mathcal{N}\|_{HS}^2 = \left|\frac{c_\gamma^2}{(i b)^{2\gamma+2}}\right| \int_\mathcal{M} \int_\mathcal{N} \left|\mathbb{J}^{a,b,c}_{d,e,\gamma}(s,t)\right|^2 t^{2\gamma+1} |s|^{2\gamma+1} dt ds.
\]
According to inequality (\ref{ModuleQFDTK}), we get

\[
\|P_\mathcal{M} Q_\mathcal{N}\|_{HS}^2 \leq \left|\frac{c_\gamma^2}{(i b)^{2\gamma+2}}\right| |\mathcal{M}||\mathcal{N}|.
\]
From classical results in functional analysis and the above inequality, we get

\[
\|P_\mathcal{M} Q_\mathcal{N}\|_{\gamma,2}^2 \leq \|P_\mathcal{M} Q_\mathcal{N}\|_{HS}^2 \leq \alpha_\gamma|\mathcal{M}||\mathcal{N}|,
\]
that is,

\begin{eqnarray}\label{QFFBTeq5.4}
\|P_\mathcal{M} Q_\mathcal{N}\|_{\gamma,2} \leq \left|\frac{c_\gamma}{(i b)^{\gamma+1}}\right| \sqrt{|\mathcal{M}||\mathcal{N}|}.
\end{eqnarray}
Finally, combining the inequalities (\ref{QFFBTeq5.3}) and (\ref{QFFBTeq5.4}) we get

\[ 
|\mathcal{M}||\mathcal{N}| \geq \frac{|b|^{2(\mu + 1)}}{c_\gamma^2} (1 - \epsilon_\mathcal{M} - \epsilon_\mathcal{N})^2.
\]

\end{proof}

In the following theorem, we provide another concentration-type uncertainty principle for quadratic-phase Fourier--Bessel transform for functions in \(L^1_{\gamma}(\mathbb{R}_+) \cap L^p_{\gamma}(\mathbb{R}_+)\).

\begin{definition}
Let \(\mathcal{M}\) and \(\mathcal{N}\) be two measurable subsets of \(\mathbb{R}_+\), and let \(h\) be a function in \(L^p_{\gamma}(\mathbb{R}_+)\), where \(1 < p \leq 2\). We say:
\begin{enumerate}
 \item \(h\) is \(\epsilon_\mathcal{M}\)  time-limited  in \(L^p_{\gamma}(\mathbb{R}_+)\)-norm if

\[
\|h - \chi_\mathcal{M} h\|_{\gamma,p} \leq \epsilon_\mathcal{M} \|h\|_{\gamma,p}.
\]

\item     \(h\) is  \(\epsilon_\mathcal{N}\)-quadratic-phase Fourier--Bessel  band-limited \(L^q_{\gamma}(\mathbb{R}_+)\)-norm, where \(q = p/(p-1)\), if

\[
\|\mathbb{B}^{a,b,c}_{d,e,\gamma}[h] - \chi_\mathcal{N} \mathbb{B}^{a,b,c}_{d,e,\gamma}[h]\|_{\gamma,q} \leq \epsilon_\mathcal{N} \|\mathbb{B}^{a,b,c}_{d,e,\gamma}[h]\|_{\gamma,q}.
\]
\end{enumerate}
\end{definition}

\begin{theorem}
Let \(\mathcal{M}\) and \(\mathcal{N}\) be two measurable subsets of \(\mathbb{R}_+\), and let \(h\) be a function in \(L^1_{\gamma}(\mathbb{R}_+) \cap L^p_{\gamma}(\mathbb{R}_+)\), where \(1 < p \leq 2\). Under the following assumptions:
\begin{enumerate}
 \item \(h\) is \(\epsilon_\mathcal{M}\)-time-limited in \(L^1_{\gamma}(\mathbb{R}_+)\)-norm,
\item  \(h\) is \(\epsilon_\mathcal{N}\)-quadratic-phase Fourier--Bessel  band-limited in \(L^q_{\gamma}(\mathbb{R}_+)\)-norm, where \(q = p/(p-1)\),
\end{enumerate}
we have

\[
\left(1 - \epsilon_\mathcal{M}\right)\left(1 - \epsilon_\mathcal{N}\right)\left\|\mathbb{B}^{a,b,c}_{d,e,\gamma}[h]\right\|_{\gamma,q} \leq \frac{c_\gamma} {|b|^{\gamma+1}}  |\mathcal{M}|^{1/p} |\mathcal{N}|^{1/q} \|h\|_{\gamma,q}.
\]
\end{theorem}

\begin{proof}
Suppose that \(|\mathcal{M}| < \infty\) and \(|\mathcal{N}| < \infty\). Let \(h \in L^1_{\gamma}(\mathbb{R}_+) \cap L^p_{\gamma}(\mathbb{R}_+)\), where \(1 < p \leq 2\). Since \(h\) is \(\epsilon_\mathcal{N}\)-quadratic-phase Fourier--Bessel  band-limited in \(L^q_{\gamma}(\mathbb{R}_+)\)-norm, we have

\begin{eqnarray*}
\left\|\mathbb{B}^{a,b,c}_{d,e,\gamma}[h]\right\|_{\gamma,q} &\leq& \epsilon_\mathcal{N} \left\|\mathbb{B}^{a,b,c}_{d,e,\gamma}[h]\right\|_{\gamma,q} + \left(\int_{\mathcal{N}} \left|\mathbb{B}^{a,b,c}_{d,e,\gamma}[h](t)\right|^q t^{2\gamma+1} dt \right)^{1/q}
\\
&\leq&  \epsilon_\mathcal{N} \left\|\mathbb{B}^{a,b,c}_{d,e,\gamma}[h]\right\|_{\gamma,q} +
 |\mathcal{N}|^{1/q} \left\|\mathbb{B}^{a,b,c}_{d,e,\gamma}[h]\right\|_{\gamma,\infty}
\end{eqnarray*}
According to the Riemann–Lebesgue lemma  of the quadratic-phase Fourier--Bessel transform (\ref{RLLQFFT}), we get

\begin{eqnarray}\label{QFFBTeq5.5}
\left\|\mathbb{B}^{a,b,c}_{d,e,\gamma}[h]\right\|_{\gamma,q} \leq \frac{c_\gamma} {|b|^{\gamma+1}} \frac{|\mathcal{N}|^{1/q}}{1 - \epsilon_\mathcal{N}} \|h\|_{\gamma,1}.
\end{eqnarray}
On the other hand, since \(h\) is \(\epsilon_\mathcal{M}\)-time-limited in \(L^1_{\gamma}(\mathbb{R}_+)\)-norm,

\begin{eqnarray*}
\|h\|_{\gamma,1} &\leq& \epsilon_{\mathcal{M}} \|h\|_{\gamma,1} + \int_{\mathcal{M}} |h(t)| t^{2\gamma+1} dt
\\
&\leq& \epsilon_{\mathcal{M}} \|h\|_{\gamma,1} + |\mathcal{M}|^{1/p} \|h\|_{\gamma,q}.
\end{eqnarray*}
Therefore,

\begin{eqnarray}\label{QFFBTeq5.6}
\|h\|_{\gamma,1} \leq \frac{|\mathcal{M}|^{1/p}}{1 - \epsilon_\mathcal{M}} \|h\|_{\gamma,q}.
\end{eqnarray}
According to inequalities (\ref{QFFBTeq5.5}) and (\ref{QFFBTeq5.6}), we obtain the desired result.

\end{proof}

\begin{remark}
The above concentration-type uncertainty principle of the quadratic-phase Fourier--Bessel transform depends on the signal \(h\). However, for the special case \(p = q = 2\), we obtain

\[
\left(1 - \epsilon_\mathcal{M}\right)\left(1 - \epsilon_\mathcal{N}\right) \leq \sqrt{\alpha_\gamma} |\mathcal{M}|^{1/2} |\mathcal{N}|^{1/2}.
\]
\end{remark}

\section*{Conclusion}
In this work, we introduced the quadratic-phase Fourier–Bessel transform and established its fundamental properties, including continuity, the Riemann–Lebesgue lemma, reversibility, and Parseval's identity. We also defined the associated translation operator and convolution product, and studied their key properties. Finally, we proved a Donoho–Stark-type uncertainty principle in the context of the quadratic-phase Fourier–Bessel transform.

\section*{Future Works}

The research presented here opens several directions that we are actively exploring:
\subsubsection*{Wavelet Transform in the Quadratic-Phase Fourier–Bessel Framework}
We are currently developing wavelet transforms associated with the quadratic-phase Fourier–Bessel transform, focusing on admissible wavelets, reconstruction formulas, and localization properties.
\subsubsection*{Uncertainty Principles}
We are extending the Donoho–Stark-type result obtained here to other forms of uncertainty principles, such as Hardy, Beurling, and Morgan-type inequalities, adapted to the quadratic-phase Fourier–Bessel domain.
\subsubsection*{Wavelet Packets}
We are investigating wavelet packet constructions in this framework, aiming to develop flexible representations, efficient decomposition schemes, and sparse signal approximations.
\subsubsection*{Calderón’s Reproducing Formula}
We are working on Calderón-type reproducing formulas in the quadratic-phase Fourier–Bessel setting to enable explicit reconstruction schemes, frame constructions, and continuous or discrete analysis.

These topics are currently under active study, and we anticipate that their full development will provide new theoretical tools and applications in harmonic analysis and signal processing within generalized Bessel domains.

\section*{Acknowledgments} The second author extends his appreciation to the Deanship of Scientific Research at Northern Border University, Arar, KSA for funding this research work under Project NO. “NBU-FFR-2025-457-xx”.

\section*{Disclosure statement}
No potential conflict of interest was reported by the author.

\bibliographystyle{abbrv}
\bibliography{bib.QFFBT}
		
	\end{document}